\documentclass{vldb}

\usepackage{amsmath,amsfonts,amssymb,amscd}
\usepackage{fullpage}
\usepackage{lastpage}
\usepackage{enumerate}
\usepackage{fancyhdr}
\usepackage{mathrsfs}
\usepackage{xcolor}
\usepackage{graphicx}
\usepackage{amsmath,bm}
\usepackage{epstopdf}
\usepackage{algorithm}
\usepackage{algorithmicx}
\usepackage{algpseudocode}

\newtheorem{definition}{Definition}

\newtheorem{theorem}{Theorem}

\newtheorem{corollary}{Corollary}

\newtheorem{lemma}{Lemma}

\newcommand{\stitle}[1]{\vspace{2mm} \noindent{\bf #1}}
\newcommand{\dom}{\ensuremath{{\cal T}}}
\renewcommand{\v}{\ensuremath{\mathbf{v}}}
\newcommand{\eat}[1]{}

\newcommand{\todo}[1]{{\color{red} #1}}

\title{On the Privacy Properties of Variants on the Sparse Vector Technique}
\numberofauthors{2}
\author{
            \alignauthor Yan Chen\\
            \affaddr{Duke University} 
                \email{yanchen@cs.duke.edu}
            \alignauthor Ashwin Machanavajjhala \\
            \affaddr{Duke University} 
                \email{ashwin@cs.duke.edu}
}

\begin{document}
\maketitle

\begin{abstract}
The sparse vector technique is a powerful differentially private primitive that allows an analyst to check whether queries in a stream are greater or lesser than a threshold. This technique has a unique property --  the algorithm works by adding noise with a finite variance to the queries and the threshold, and guarantees privacy that only degrades with (a) the maximum sensitivity of any one query in stream, and (b) the number of positive answers output by the algorithm. Recent work has developed variants of this algorithm, which we call {\em generalized private threshold testing}, and are claimed to have privacy guarantees that do not depend on the number of positive or negative answers output by the algorithm. These algorithms result in a significant improvement in utility over the sparse vector technique for a given privacy budget, and have found applications in frequent itemset mining, feature selection in machine learning and generating synthetic data. 

In this paper we critically analyze the privacy properties of generalized private threshold testing. We show that  generalized private threshold testing does not satisfy $\epsilon$-differential privacy  for any finite $\epsilon$. We identify a subtle error in the privacy analysis of this technique in prior work. Moreover, we show an adversary can use generalized private threshold testing to recover counts from the datasets (especially small counts) exactly with high accuracy, and thus can result in individuals being reidentified. We demonstrate our attacks empirically on real datasets. 
\end{abstract}

\section{Introduction}\label{sec:intro}

A popular building block for $\epsilon$-differentially private query answering is the Laplace mechanism. Given a set of queries ${\cal Q}$ as input, the Laplace mechanism adds noise drawn independently from the Laplace distribution to each query in ${\cal Q}$. Adding noise with standard deviation of $\sqrt{2}/\epsilon$ to each of the queries in ${\cal Q}$ ensures $(\Delta_{\cal Q} \cdot \epsilon)$-differential privacy, where $\Delta_{\cal Q}$ is the {\em sensitivity} of ${\cal Q}$, or the sum of the changes in each of the queries $Q \in {\cal Q}$ when one row is added or removed from the input database. Increasing the number of queries increases the sensitivity, and thus for a fixed privacy budget the  mechanism's accuracy is poor for large sets of queries (unless the queries operate on disjoint subsets of the domain). 


The {\em sparse vector technique} (SVT) \cite{SVT} is an algorithm that allows testing whether a stream of queries is greater or lesser than a threshold $\tau$. SVT works by adding noise to both the threshold $\tau$ and to each of the queries $Q \in {\cal Q}$. If noise with standard deviation of $\sqrt{2}/\epsilon$ is added to the threshold and each of the queries, SVT can be shown to satisfy $c\Delta \epsilon$-differential privacy, where $\Delta$ is the maximum sensitivity of any single query in ${\cal Q}$ and $c$ is the number of positive answers (greater than threshold) that the algorithm outputs. Note that the privacy guarantee does not depend on the number of queries with negative answers, and the sensitivity does not necessarily increase with an increase in number of queries.

Recent work has explored the possibility of extending this technique to eliminate the dependence on the number of positive answers ($c$).  We call this idea {\em generalized private threshold testing}, and it works like SVT -- noise is added to both the threshold and each of the queries using noise whose standard deviation only depends only on $\epsilon$ and maximum sensitivity $\Delta$ of a single query in ${\cal Q}$.  Generalized private threshold testing has been claimed to ensure differential privacy with the privacy parameter having no dependence on the number of positive or negative queries! Hence, generalized private threshold testing has been used to develop algorithms with high utility for  private frequent itemset mining \cite{NoiseCut}, feature selection in private classification \cite{EML} and generating synthetic data \cite{KDD15}.

In this article, we critically analyze the privacy properties of generalized private threshold testing. We make the following contributions: 
\begin{itemize}
\item We show that generalized private threshold testing does not satisfy $\epsilon$-differential privacy, where $\epsilon$ does not depend on the number of queries being tested. We identify a specific claim in the privacy analysis in prior work that is assumed to hold, but does not in reality.
\item We show specific examples of neighboring databases, queries and outputs that violate the requirement that the algorithm output is insensitive to adding or removing a row in the data. 
\item We display an attack algorithm and demonstrate that using generalized private threshold testing could make it possible for adversaries to reconstruct the counts for each cell with high probability, especially the cells with small counts.
\end{itemize}

\noindent{\bf Organization:}
Section~\ref{sec:prelim} surveys concepts on differential privacy and the sparse vector technique. In Section~\ref{sec:GPTT}, we introduce generalized private threshold testing  and its instantiations in prior works. We show that generalized private threshold testing  does not satisfy differential privacy in Section~\ref {sec:nonprivate}. We describe an attack algorithm for reconstructing the counts of cells in the input datasets by using generalized private threshold testing in Section~\ref{sec:recon} and demonstrate our attacks on real datasets.


\section{Preliminaries}\label{sec:prelim}
\stitle{Databases:}
A database $D$ is a multiset of entries whose values come from a domain $\dom = \{u_1, u_2, \ldots, u_k\}$. Let $n = |D|$ denote the number of entries in the database. We represent a database $D$ as a histogram of counts over the domain. That is, $D$ is represented as a vector $\mathbf{x} \in \mathbb{N}^k$ where $\mathbf{x}[i]$ or $x_i$ denotes the true count of entries in $D$ with the $i^{\rm th}$ value of the domain $\dom$.

\stitle{Differential Privacy:}
We define a neighborhood relation $N$ on databases as follows: Two databases $D_{1}$ and $D_{2}$  are considered neighboring datasets if and only if they differ in the presence or absence of a single entry. That is, $(D_1, D_2) \in N$ iff for some $t\in \dom$, $D_{1} = D_{2} \cup \{t\}$ or $D_{2} = D_{1} \cup \{t\}$. Equivalently, if $\mathbf{x}_1$ and $\mathbf{x}_2$ are histograms of neighboring databases, $||\mathbf{x}_1 - \mathbf{x}_2||_1 = 1$. An algorithm satisfies differential privacy if its outputs are statistically similar on neighboring databases. 

\begin{definition}[$\epsilon$-differential privacy]
A randomized algorithm $\mathcal{M}$ satisfies $\epsilon$-differential privacy if for any pair of neighboring databases $(D_1, D_2) \in N$, and $\forall S \in$ range($\mathcal{M}$),
\begin{equation}
Pr[\mathcal{M}(D_{1}) = S] \leq e^{\epsilon} \cdot P[\mathcal{M}(D_{2}) = S]
\end{equation}
\end{definition}

The value of $\epsilon$, called \emph{privacy budget}, controls the level of the privacy, and limits how much an adversary can distinguish one dataset with its neighboring datasets given an output. Smaller $\epsilon$'s correspond to more privacy. 

Differentially private algorithms satisfy the following {\em composition} properties. Suppose $M_1(\cdot)$ and $M_2(\cdot)$ be $\epsilon_1$- and $\epsilon_2$-differentially private algorithms.
\begin{itemize}
\item {\em Sequential Compositon:} Releasing the outputs of $M_1(D)$ and $M_2(D)$ satisfies $\epsilon_1+\epsilon_2$-differential privacy. 
\item {\em Parallel Composition:} Releasing $M_1(D_1)$ and $M_2(D_2)$, where $D_1 \cap D_2 = \emptyset$ satisfies $\max(\epsilon_1, \epsilon_2)$-differential privacy.
\item {\em Postprocessing}: For any algorithm $M_3(\cdot)$, releasing $M_3(M_1(D))$ still satisfies $\epsilon_1$-differential privacy. That is, postprocessing an output of a differentially private algorithm does not incur any additional loss of privacy.
\end{itemize}

Thus, complex differentially private algorithms can be build by composing simpler private algorithms. \emph{Laplace Mechanism} \cite{DP} is one such widely used building that achieves differential privacy that adds noise from a Laplace distribution with a scale proportional to the \emph{global sensitivity}.

\begin{definition}[Global Sensitivity]
The global sensitivity of a fuction $f:{\cal D} \rightarrow \mathbb{R}^{n}$, denoted as $\Delta(f)$, is defined to be the maximum $L_{1}$ distance of the output from any two neighboring datasets $D_{1}$ and $D_{2}$.
\begin{equation}
\Delta(f) = \max_{(D_{1},D_{2}) \in N} ||f(D_{1}) - f(D_{2})||_{1}
\end{equation} 
\end{definition}

\begin{definition}[Laplace Mechanism]
For any function $f:{\cal D} \rightarrow \mathbb{R}^{d}$, the Laplace Mechanism $\mathcal{M}$ is given by: $\mathcal{M}$(D) = f(D) + $\eta$. $\eta$ is a vector of independent random variables drawn from a Laplace distribution with the probability density fuction $p(x| \lambda) = \frac{1}{2\lambda}e^{-|x|/\lambda}$, where $\lambda = \Delta(f) / \epsilon$.
\end{definition}

\begin{theorem}
Laplace Mechanism satisfies $\epsilon$-differential privacy. 
\end{theorem}

\stitle{Sparse Vector Technique}:

\begin{algorithm}[t]
\caption{Sparse Vector Technique}
{\bf Input}: Dataset $D$,  a stream of queries $q_1,q_2, \dots$ with bounded sensitivity $\Delta$, threshold $\theta$, a cutoff point $c$ and privacy budget $\epsilon$\\
{\bf Output}: a stream of answers
\begin{algorithmic}[1]
	\State $\tilde{\theta} \gets \theta + Lap(2\Delta/\epsilon)$, $count \gets 0$
	\For{each query $i$}
		\State $v_i \gets Lap(2\Delta*c/\epsilon)$
		\If {$q_i(D) + v_i \ge \tilde{\theta}$}
			\State Output $v_i = \top$
			\State $count \gets count+1$
		\Else
			\State Output $v_i = \bot$
		\EndIf
		\If {$count \ge c$}
			\State Abort
		\EndIf
	\EndFor
\end{algorithmic}
\label{alg:SVT}
\end{algorithm}

Algorithm~\ref{alg:SVT} shows the details of the sparse vector technique (SVT). The input of SVT is a stream of queries ${\cal Q} = \{q_1, q_2, \dots, q_k\}$, where each query $q \in {\cal Q}$ has sensitivity bounded by $\Delta$, a threshold $\theta$, and a limit $c$. For every query, SVT outputs either $\bot$ (negative response) or $\top$ (positive response). SVT works in two steps: (1) Perturb the threshold $\theta$ by adding noise drawn from the Laplace distribution with scale $\frac{2}{\epsilon}$, getting $\tilde{\theta}$. (2) Perturb each query $q_i$ adding Laplace noise ($Lap(\frac{2c}{\epsilon \Delta})$) getting $\tilde{q}$. Output $\bot$ if $\tilde{q_i} < \tilde{\theta}$ and $\top$ otherwise. This algorithm stops when it outputs $c$ positive responses.

\begin{theorem}\cite{privbook}
SVT satisfies $\epsilon$-differential privacy. Moreover if the number of queries is $k$ and their max sensitivity $\Delta$, with probability at least $1-\delta$, for every $a_i = \top$, $q_i > \tau - \alpha$ and for every $a_i = \bot$, $q_i < \tau + \alpha$, where
\[\alpha = O(\frac{c\Delta}{\epsilon} \cdot (\log k + \log(2/\delta))\]
\end{theorem}

\section{Generalized Private Threshold Testing} \label{sec:GPTT}

In this section, we describe a method called Generalized Private Threshold Testing (GPTT) (see Algorithm~\ref{alg:GPTT}) that generalizes variations of the sparse vector technique that do not require a limit on the number of positive (or negative) responses.  


GPTT takes as input a dataset $D$, a set of queries $Q=\{q_1,\dots,q_n\}$ with bounded sensitivity $\Delta$, threshold $\theta$ and a privacy budget $\epsilon$. For every query GPTT  outputs either $\bot$ or $\top$ that approximates whether or not the queries are smaller than the threshold. GPTT works exactly like SVT -- the threshold is perturbed using noise drawn from $Lap(\Delta/\epsilon_1)$ and the queries are perturbed using noise from $Lap(\Delta/\epsilon_2)$, and the output is computed by comparing the noisy query answer with the noisy threshold. The only difference is that there is no limit on the number of positive or negative queries. 

GPTT is a generalization of variations presented in prior work. 
Lee and Clifton \cite{NoiseCut} used GPTT for private frequent itemset mining with $\epsilon_{1} = \frac{\epsilon}{4}$ and $\epsilon_2 = \frac{3\epsilon}{4}$. Chen et al \cite{KDD15} instantiate GPTT with  $\epsilon_1 = \epsilon_2 = \frac{\epsilon}{2}$, for generating synthetic data. Stoddard et al \cite{EML} observed that the privacy guarantee does not depend on $\epsilon_2$ and propose the Private Threshold Testing algorithm that is identical to GPTT with $\epsilon_1 = \epsilon$ and $\epsilon_2 = \infty$. Private threshold testing was used for private feature selection for classification. 

\subsection{Privacy Analysis of GPTT}
We now extend the privacy analysis from prior work \cite{NoiseCut, KDD15, EML} to generalized private threshold testing. We will show in the next section that this privacy analysis is flawed and that GPTT does not satisfy differential privacy. 

\begin{algorithm}[t]
\caption{Generalized Private Threshold Testing}
Input: Dataset $D$,  a set of queries $Q=\{q_1,\dots,q_n\}$ with bounded sensitivity $\Delta$, threshold $\theta$, privacy parameters $\epsilon_1$, $\epsilon_2$\\
Output: A vector of answers $\mathbf{v} = [v_1, v_2, \ldots, v_n]\in \{\bot, \top\}^n$
\begin{algorithmic}[1]
	\State $\tilde{\theta} \gets \theta + Lap(\Delta/\epsilon_1)$
	\For{ $q_i \in Q$}
		\State $\tilde{q}_i \gets q_i(D) + Lap(\Delta/\epsilon_2)$
		\If {$\tilde{q}_i < \tilde{\theta}$}
			\State $v_i \gets \bot$
		\Else
			\State $v_i \gets \top$
		\EndIf
	\EndFor
	\State\Return $\mathbf{v}$
\end{algorithmic}
\label{alg:GPTT}
\end{algorithm}
 
Given any set of queries $Q=\{q_1,\dots,q_n\}$, let the vector $\mathbf{v} = <v_1,\dots,v_n> \in \{\bot, \top\}^n$ denote the output of GPTT. Given any two neighbering databases $D_1$ and $D_2$, let $V_1$ and $V_2$ denote the output distribution on $\mathbf{v}$ when $D_1$ and $D_2$ are the input databases, respectively. We use $\mathbf{v}^{<t}$ to denote $t-1$ previous answers(i.e., $\mathbf{v}^{<t}=<v_1,\dots,v_{t-1}>$). Then we have
\begin{eqnarray*}
\frac{V_1(\v)}{V_2(\v)} &=& \frac{\prod_{i=1}^{n}V_1(v_i =a_i \mid \v^{<i})}{\prod_{i=1}^{n}V_2(v_i =a_i \mid \v^{<i})} \\
&=& \prod_{i:a_i=\top}\frac{V_1(v_i =\top \mid \mathbf{v}^{<i})}{V_2(v_i =\top \mid \mathbf{v}^{<i})}\prod_{i:a_i=\bot}\frac{V_1(v_i =\bot \mid \mathbf{v}^{<i})}{V_2(v_i =\bot \mid \mathbf{v}^{<i})}
\end{eqnarray*}

Let $H_i$(x) be the probability that $q_i$ is positive (i.e., $v_i = \top$) in $D$ when the noisy threshold is $x$. That is, 
\begin{eqnarray*}
H_i(x)  = P[v_i = \top \mid \tilde{\theta} = x, \mathbf{v}^{<i}]
\end{eqnarray*}
Then, given a specific noisy threshold $\tilde{\theta} = x$, the probability that $v_i = \top$ is independent of the answers to previous queries. That is, 
\begin{equation}
H_i(x)  \ =\  P[v_i = \top \mid x, \mathbf{v}^{<i}] = P[v_i = \top \mid x] 
\end{equation}
Thus, if $f(y; \mu, \lambda) =  \frac{1}{2\lambda}exp(-\frac{|y-\mu|}{\lambda})$, then
\begin{eqnarray*}
H_i(x)  &=& \int_{x}^{\infty} f(y; q_i, \frac{\Delta}{\epsilon_2})dy \ = \ \int_{x+\Delta}^{\infty} f(y; q_i + \Delta, \frac{\Delta}{\epsilon_2})dy
\end{eqnarray*}



Prior work uses the above property of $H_i(x)$ to show that GPTT satisfies $2\epsilon_1$-differential privacy.

Let $S = \{i \mid a_i = \top~and~q_i(D_1) = q_i(D_2)\}$ and $\bar{S} = \{i \mid a_i = \top~and~q_i(D_1) \neq q_i(D_2)\}$.
\begin{eqnarray*}
\prod_{i:a_i=\top}V_1(v_i =\top \mid \mathbf{v}^{<i}) = \prod_{i \in S}V_1(v_i =\top \mid \mathbf{v}^{<i})\prod_{i \in \bar{S}}V_1(v_i =\top \mid \mathbf{v}^{<i})
\end{eqnarray*}
Let $H_i^{1}(x)$ and $H_i^{2}(x)$ denote the probability that $v_i = \top$ in $D_1$ and $D_2$, resp., when the noisy threshold is $x$. Then we have, 
\begin{eqnarray*}
&&\prod_{i \in S}V_1(v_i =\top \mid \mathbf{v}^{<i})\\
&=& \int_{-\infty}^{\infty} P[\tilde{\theta} = x] \prod_{i \in S} H_i^{1}(x)dx  \\
&=& \int_{-\infty}^{\infty} P[\tilde{\theta} = x] \prod_{i \in S} H_i^{2}(x)dx \mbox{\hspace{5mm} since $q_i(D_1) = q_i(D_2)$}\\
& =& \prod_{i \in S}V_2(v_i =\top \mid \mathbf{v}^{<i}) 
\end{eqnarray*}
\begin{eqnarray*}
&&\prod_{i \in \bar{S}}V_1(v_i =\top \mid \mathbf{v}^{<i})\\
&=& \int_{-\infty}^{\infty} P[\tilde{\theta} = x] \prod_{i \in \bar{S}} H_i^{1}(x)dx\\
&\leq& exp(\epsilon_1)\int_{-\infty}^{\infty}P[\tilde{\theta} = x - \Delta]\prod_{i \in \bar{S}} H_i^{2}(x-\Delta)dx\\
&=& \prod_{i \in \bar{S}}V_2(v_i =\top \mid \mathbf{v}^{<i})
\end{eqnarray*}

Thus, it is seen that 
\begin{eqnarray*}
\prod_{i:a_i=\top}V_1(v_i =\top \mid \mathbf{v}^{<i}) \leq exp(\epsilon_1) \prod_{i:a_i=\top}V_2(v_i =\top \mid \mathbf{v}^{<i})
\end{eqnarray*}
Similarly, 
\begin{eqnarray*}
\prod_{i:a_i=\bot}V_1(v_i =\bot \mid \mathbf{v}^{<i}) \leq exp(\epsilon_1) \prod_{i:a_i=\bot}V_2(v_i =\bot \mid \mathbf{v}^{<i})
\end{eqnarray*}
Therefore, $\frac{V_1(v)}{V_2(v)} \leq exp(2\epsilon_1)$

\stitle{Remarks:} Note that adding noise to the queries is not really required. The above proof will go through even if $\epsilon_2 = \infty$. In fact as we will see next, we can achieve the same utility no matter what the value of $\epsilon_2$ is. 

\subsection{Utility of GPTT}
We first consider the utility of the case when $\epsilon_2 = \infty$ (i.e. no noise added to queries), and show that we can achieve (almost) the same utility even when $\epsilon_2$ is finite. 

\begin{theorem}
\label{theo:PTT}
For GPTT with parameters $\epsilon_1$ and  $\epsilon_2 = \infty$, with the probability at least $1-\delta$, $v_i = \bot$ implies $q_i < \theta + \alpha$ and $v_i = \top$ implies $q_i > \theta - \alpha$, where
\begin{eqnarray*}
\alpha = \frac{\Delta}{\epsilon_1}\log(\frac{1}{\delta})
\end{eqnarray*}
where $\Delta$ is the max sensitivity of input queries.
\end{theorem}
\begin{proof}
All we need to show is that the noise added to the threshold is at most $\pm \alpha$  with probability $1-\delta$. Since $\tilde{\theta} = \theta + Lap(\frac{\Delta}{\epsilon_1})$, we have
\begin{eqnarray*}
\lefteqn{P(|\tilde{\theta} - \theta|) < \alpha)  \ge 1-\delta }\\
&\Rightarrow& P(-\alpha<Lap(\frac{\Delta}{\epsilon_1})<\alpha) \ge 1-\delta \\
&\Rightarrow& P(Lap(\frac{\Delta}{\epsilon_1}) > \alpha) \leq \frac{\delta}{2}\\
&\Rightarrow& 1 - (1-\frac{1}{2} exp(-\frac{\epsilon_1 \alpha}{\Delta})) \leq \frac{\delta}{2}\\
&\Rightarrow& \alpha \ge \frac{\Delta}{\epsilon_1} \log(\frac{1}{\delta})
\end{eqnarray*}
\end{proof}

Now we extend the utility for GPTT when $\epsilon_2 < \infty$.

\begin{theorem}
\label{theo:GPTT}
Let $D$ be a database and $Q$ a query set with maximum sensitivity of $\Delta$. For every $\beta, \delta > 0$, we can use GPTT with parameters $\epsilon_1, \epsilon_2 < \infty$ to determine whether $q_i < \theta + \alpha$ or $q_i > \theta - \alpha$ for any $q_i \in Q$, with probability $(1-\delta)(1 - \beta)$, where
\begin{eqnarray*}
\alpha = \frac{\Delta}{\epsilon_1}\log(\frac{1}{\delta})
\end{eqnarray*}
\end{theorem}
\begin{proof}
Since the privacy of GPTT does not depend on the number of queries $Q$ (as long as sensitivity is bounded by $\Delta$), we can consider a new query set $Q'$ that has $t$ copies $\{q_{i1}, q_{i2}, \ldots, q_{it}\}$ of each query $q_i \in Q$. Then for each query $q_i \in Q$, we have t independent comparisons of the noisy query answer $\tilde{q}_{ij}$ and the noisy threshold $\tilde{\theta}$. We use the majority of these $t$ results to determine whether $q_i$ is smaller or greater than $\tilde{\theta}$. 

We can show that with probability at least $1-\beta$, we can correctly identify whether $q_i(D)$ is greater or lesser than the noisy threshold $\tilde{\theta}$.

Without loss of generality, suppose $q_i < \tilde{\theta}$, then $p = P(\tilde{q_i} < \tilde{\theta}) > \frac{1}{2}$. Let $\{X_j, j \ge 1\}$ be an sequence of i.i.d. binary random variables with expected $E[X_j]=p$, where $X_j$ is $1$ if $\tilde{q}_{ij} < \tilde{\theta}$ and $0$ otherwise. Based on the law of large numbers, for any positive number $\gamma$ we have

\begin{eqnarray*}
&&\underset{t \to \infty}{lim} Pr(|\frac{1}{t}(X_1+\dots+X_t) - p| > \gamma) = 0\\
&\Rightarrow& \underset{t \to \infty}{lim} Pr(X_1+\dots+X_t > \frac{t}{2}) = 1
\end{eqnarray*}

Thus, there exists a number $t$ s.t. for every $q_i \in Q$, we can determine whether $q_i$ is smaller or greater than $\tilde{\theta}$ with probability greater than $1-\frac{\beta}{|Q|}$. So all the $q_i$ will be correctly judged with probability equals to $ (1-\frac{\beta}{|Q|})^{|Q|} \approx 1-\beta$.

We get the desired result by now combining with the proof of Theorem~\ref{theo:PTT}.
\end{proof}

We can see that the information leaked by GPTT with $\epsilon_2 < \infty$ will tend to the information leaked by GPTT with $\epsilon_2 = \infty$ as the number of copies $t$ goes to infinity. Thus, we can just focus on the case when $\epsilon_2 = \infty$.

\section{GPTT is not  Private}\label{sec:nonprivate}
While prior work claims that GPTT is indeed differentially private (as discussed in Section~\ref{sec:GPTT}), we show that this algorithm does not satisfy the privacy condition. The proof is constructive and will show examples of pairs of  neighboring datasets and queries for which GPTT violates differential privacy.

\begin{theorem}\label{thm:GPTTnonprivate}
GPTT does not satisfy $\epsilon$-differential privacy for any finite $\epsilon$. 
\end{theorem}

In the proof of this theorem, we start with the case of $\epsilon_2 = \infty$, where the true query answer is compared with the noisy threshold. It is easy to show that GPTT does not satisfy differential privacy in this case, since deterministic information about the queries is leaked. In particular, if $v_i = \bot$ and $v_j = \top$, we are certain that on the input database $D$, there is some $x$ such that $q_i(D) < x \leq q_j(D)$. 

The proof of the more general case follows from Theorem~\ref{theo:GPTT}  which shows that anything that is disclosed by GPTT with $\epsilon_2 = \infty$ is also disclosed with high probability (by making sufficient number of copies of the input queries). We present the formal proof below. 
\begin{proof}
Consider two queries $q_1$ and $q_2$ with sensitivity $1$. For special case of GPTT, where $\epsilon_2 = \infty$, suppose in dataset $D$, $q_1(D) = 0$ and  $q_2(D) = 1$. Also suppose in a neighboring dataset $D'$, $q_1(D') = 1$ and  $q_2(D') = 0$. Let the threshold $\theta$ be $0$. Then, the probability of getting an output $v_1 = \bot$ and $v_2 =\top$ is $>0$ under database $D$; this corresponds to the probability that the noisy threshold is within $(0,1)$. However, under the neighboring dataset $D'$, $P[v_1 = \bot, v_2 = \top] = 0$. This is because $q_1(D') > q_2(D')$. Thus for any noisy threshold $\tilde{\theta}$,  
\[ v_1 = \bot \implies q_1(D') < \tilde{\theta} \implies q_2(D') < \tilde{\theta} \implies v_2 = \bot\]
Hence, GPTT with $\epsilon_2 = \infty$ does not satisfy differential privacy. 

To prove that GPTT does not satisfy differential privacy when $\epsilon_2 < \infty$, we construct a similar counterexample as above, except that we use $t$ copies of $q_1$ and $q_2$. So again, let the query set $Q = \{q_1, \ldots, q_{2t}\}$ be such that on dataset $D$, $q_1(D) = \ldots = q_t(D) = 0$ and $q_{t+1}(D) = \ldots = q_{2t}(D) = 1$. We assume all queries have sensitivity 1. On neighboring dataset $D'$, $q_1(D) = \ldots = q_t(D) =1$ and $q_{t+1}(D) = \ldots = q_{2t}(D) = 0$. Let the threshold $\theta = 0$. Let the output vector $\mathbf{v}$ be such that $v_1 = \ldots = v_t = \bot$ and $v_{t+1} = \ldots = v_{2t} = \top$. Then we have: 
\begin{eqnarray*}
\lefteqn{V(\mathbf{v})  = P[GPTT(D) = \mathbf{v}]}\\
&=& \int_{-\infty}^{\infty} P(\tilde{\theta} = z)\prod_{i=1}^{t}P(\tilde{q_i}<z) \prod_{i=t+1}^{2t}P(\tilde{q_i}\ge z) dz\\
&=& \int_{-\infty}^{\infty}f_{\epsilon_1}(z)(F_{\epsilon_2}(z)(1-F_{\epsilon_2}(z-1)))^{t} dz \\
&=& \int_{-\infty}^{\infty}f_{\epsilon_1}(z)(F_{\epsilon_2}(z)-F_{\epsilon_2}(z)F_{\epsilon_2}(z-1))^{t} dz 
\end{eqnarray*}
where $f_{\epsilon}(z)$ and $F_{\epsilon}(z)$ are the pdf and cdf respectively of the Laplace distribution with parameter $1/\epsilon$. 
Similarly, we have on the neighboring database $D'$, 
\begin{eqnarray*}
\lefteqn{V'(\mathbf{v})  = P[GPTT(D') = \mathbf{v}]}\\
& = & \int_{-\infty}^{\infty}f_{\epsilon_1}(z)(F_{\epsilon_2}(z-1)-F_{\epsilon_2}(z)F_{\epsilon_2}(z-1))^{t} dz 
\end{eqnarray*}

Let  $V'(\mathbf{v}) = \alpha$ and let $\delta = |F_{\epsilon_1}^{-1}(\frac{\alpha}{4})|$. Since $\alpha \leq 1$, $\delta$ is greater than the $75^{th}$ percentile of a Laplace distribution with scale $1/\epsilon_1$. That is,  
\[\frac{\alpha}{2} \ = \ \int_{-\infty}^{-\delta} f_{\epsilon_1}(t) dt + \int_{\delta}^{\infty} f_{\epsilon_1}(t) dt \]
Moreover, note that since $F_{\epsilon_2}(z-1) < F_{\epsilon_2}(z)$ for all $z$, we have 
\[1 < \kappa(z)  = \frac{F_{\epsilon_2}(z)-F_{\epsilon_2}(z)F_{\epsilon_2}(z-1)} {F_{\epsilon_2}(z-1)-F_{\epsilon_2}(z)F_{\epsilon_2}(z-1)}\]
Let $\kappa$ denote the minimum value $\kappa(z)$ takes over all $z \in [-\delta, \delta]$; thus $\kappa > 1$. 

Now we get 
\begin{eqnarray*}
\lefteqn{V'(\mathbf{v}) = \alpha }\\
&=&\int_{-\infty}^{\infty}f_{\epsilon_1}(z)(F_{\epsilon_2}(z-1)-F_{\epsilon_2}(z)F_{\epsilon_2}(z-1))^{t} dz \\
&<& \int_{z \not\in[-\delta,\delta]} f_{\epsilon_1}(z) dz + \int_{-\delta}^{\delta}f_{\epsilon_1}(z)(F_{\epsilon_2}(z-1)-F_{\epsilon_2}(z)F_{\epsilon_2}(z-1))^{t} dz \\
&=& \frac{\alpha}{2} + \int_{-\delta}^{\delta}f_{\epsilon_1}(z)(F_{\epsilon_2}(z-1)-F_{\epsilon_2}(z)F_{\epsilon_2}(z-1))^{t} dt
\end{eqnarray*}
Thus, we have 
\begin{eqnarray*}
\lefteqn{\int_{-\delta}^{\delta}f_{\epsilon_1}(z)(F_{\epsilon_2}(z-1)-F_{\epsilon_2}(z)F_{\epsilon_2}(z-1))^{t} dz}\\
&>& \frac{1}{2} \int_{-\infty}^{\infty}f_{\epsilon_1}(z)(F_{\epsilon_2}(z-1)-F_{\epsilon_2}(z)F_{\epsilon_2}(z-1))^{t} dz
\end{eqnarray*}
Therefore, 
\begin{eqnarray*}
V(\mathbf{v}) &=& \int_{-\infty}^{\infty}f_{\epsilon_1}(z)(F_{\epsilon_2}(z)-F_{\epsilon_2}(z)F_{\epsilon_2}(z-1))^{t} dz \\
&>& \int_{-\delta}^{\delta}f_{\epsilon_1}(z)(F_{\epsilon_2}(z)-F_{\epsilon_2}(z)F_{\epsilon_2}(z-1))^{t} dz\\
&>& \int_{-\delta}^{\delta}f_{\epsilon_1}(t)(F_{\epsilon_2}(z-1)-F_{\epsilon_2}(z)F_{\epsilon_2}(z-1))^{t} \kappa^t dz \\
&>& \frac{\kappa^t}{2} \int_{-\infty}^{\infty}f_{\epsilon_1}(z)(F_{\epsilon_2}(z-1)-F_{\epsilon_2}(z)F_{\epsilon_2}(z-1))^{t} dz \\
&=& \frac{\kappa^t}{2} V'(\mathbf{v})
\end{eqnarray*}

Since $\kappa > 1$, for every $\epsilon > 1$ there exists a $t$ such that $V(\mathbf{v}) > e^{\epsilon} V'(\mathbf{v})$ which violates differential privacy.
\end{proof}

\clearpage
\subsection{Intuition}
We believe that there is a subtle error in the privacy analysis in prior work (discussed in Section~\ref{sec:GPTT}). Prior work splits the probability $V(\mathbf{v})$ as: 
\begin{eqnarray*}
V(\mathbf{v}) & = & \prod_{i:v_i = \bot} P[v_i = \bot | \mathbf{v}^{<i}] \prod_{i:v_i = \top} P[v_i = \top | \mathbf{v}^{<i}] \\ 
& = & \int f(x) \prod_{i:v_i = \bot} P[v_i = \bot | x, \mathbf{v}^{<i}]dx\\
&& \times \int f(x) \prod_{i:v_i = \top} P[v_i = \top | x, \mathbf{v}^{<i}]dx\\
& = & \int f(x) \prod_{i:v_i = \bot} P[v_i = \bot | x]dx \\
&& \times \int f(x) \prod_{i:v_i = \top} P[v_i = \top | x]dx
\end{eqnarray*}
where $f(x) = P[\tilde{\theta} = x]$. 
This decomposition into $\top$ and $\bot$ answers are wrong. The main problem comes from the fact that it uses the unconditioned $f(x)$ for all queries, but the distribution of the noisy threshold would be affected given the previous output. To take a simple example, let $q_1 = m > 0$, $q_2 = 0$, $\theta = 0$ and assume $v_1 = \bot$, $v_2 = \top$. For ease of explanation assume $\epsilon_2 = \infty$ (but the argument would work for finite $\epsilon_2$ as well). Now we can compute the probability of GPTT outputing $v_1$, $v_2$ as
\begin{eqnarray*}
P(v_1 = \bot, v_2 = \top)& =& P(v_1 = \bot)P(v_1 = \bot \mid v_2 = \top) \\ 
&=& P(v_1 = \bot)P(m < \tilde{\theta} \mid 0 \ge \tilde{\theta}) = 0
\end{eqnarray*}

However, if we use the expression above, we have
\begin{eqnarray*}
V(\mathbf{v}) & = & \prod_{i:v_i = \bot} P[v_i = \bot | \mathbf{v}^{<i}] \prod_{i:v_i = \top} P[v_i = \top | \mathbf{v}^{<i}] \\
&=& \int f(x) P(v_1 = \bot \mid x)dx \times \int f(x) P(v_2 = \top \mid x)dx \\
&=& \int f(x) P(m < x)dx \times \int f(x) P(0 > x) dx \\
&=& F(m) \times (1 - F(0)) > 0
\end{eqnarray*}
where $F(x)$ is the distribution function of the noisy threshold.

Actually, the right decomposition should be 
\begin{eqnarray*}
V(\mathbf{v}) & = & \prod_{i:v_i = \bot} P[v_i = \bot | \mathbf{v}^{<i}] \prod_{i:v_i = \top} P[v_i = \top | \mathbf{v}^{<i}] \\ 
&=& \prod_{i:v_i = \bot} \int f(x | \mathbf{v}^{<i} ) P[v_i = \bot | x, \mathbf{v}^{<i}]dx\\
&& \times \prod_{i:v_i = \top} \int f(x | \mathbf{v}^{<i} )  P[v_i = \top | x, \mathbf{v}^{<i}]dx\\
&=& \prod_{i:v_i = \bot} \int f(x | \mathbf{v}^{<i} ) P[v_i = \bot | x]dx\\
&& \times \prod_{i:v_i = \top} \int f(x | \mathbf{v}^{<i} )  P[v_i = \top | x]dx
\end{eqnarray*}

\section{Reconstructing the data using GPTT}\label{sec:recon}
In the last section, we showed that generalized private threshold testing does not satisfy $\epsilon$-differential privacy for any $\epsilon$. While this is an interesting result, it still leaves open whether GPTT indeed leaks a significant amount of information from the dataset, and allows attacks like re-identification of individuals based on quasi-identifiers. In this section, we answer this question in the affirmative, and show that generalized private threshold testing may disclose the exact counts of domain values. Exact disclosure of cells with small counts (especially, cells with counts 0, 1 and 2) reveal the presence of unique individuals in the data who can be susceptible to reidentification attacks. 

We will show our attack for the special case of GPTT where $\epsilon_2 = \infty$. Since GPTT with finite $\epsilon_2$ can be made to leak as much information about a set of queries as GPTT with $\epsilon_2 = \infty$ (Theorem~\ref{theo:GPTT}) with high probability, we will not separately consider that case. 

In our attack, we use a set of {\em difference} queries that compute the difference between the counts of pairs of domain elements. 
\begin{definition}[Difference Query]
Let $u_1, u_2 \in \dom$ be a pair of domain elements, and let $x_1$ and $x_2$ be their counts in a dataset $D$. The difference query  $diff(u_1, u_2)$ is given by 
\[diff(u_1, u_2) \ = \ x_1 - x_2\]
\end{definition}
Note that each $diff(u,v)$ query has sensitivity $\Delta = 1$.

\begin{algorithm}[t]
\caption{Attack Algorithm}
Input: Dataset $D$ with domain $\dom$,  privacy param $\epsilon$\\
Output: A partitioning of the domain $P = \{P_1, \ldots, P_p\}$ 
\begin{algorithmic}[1]
	\State Let ${\cal Q} \gets \{diff(x_u, x_v) \ \mid \ \forall u,v \in \dom\}$
	\State Let $\theta \gets \lceil \frac{1}{\epsilon} \log \left(\frac{1}{\delta}\right)\rceil$ 
	\State Run GPTT($\epsilon_1 = \epsilon$, $\epsilon_2 = \infty$) on database $D$, queries ${\cal Q}$ and threshold $\theta$.
	\State $\forall v \in \dom$,  Let $larger(v)$ be the set $\{u \in \dom \ \mid\  \mbox{GPTT outputs $\top$ for } diff(x_u,x_v) \}$
	\State Construct {\em ordered} partition of the domain ${\cal P} = \{P_1, \ldots, P_p\}$, such that
	\Statex \hspace{5mm} $\forall u, v \in P_i$, $larger(u) = larger(v)$, and
	\Statex \hspace{5mm} $\forall u \in P_i, v \in P_{i+1}$, $larger(v) \subsetneq larger(u)$
	\State \Return $P$
\end{algorithmic}
\label{alg:attack}
\end{algorithm}

Our attack algorithm is defined in Algorithm~\ref{alg:attack}. Given the input dataset $D$ with domain $\dom$, we apply GPTT $(\epsilon_2 = \infty)$ to the set of all difference queries using all pairs of domain elements $u, v \in \dom$. We use a threshold $\theta = \lceil \frac{1}{\epsilon} \log \left(\frac{1}{\delta}\right)\rceil$. Pairs of domain elements $u, v \in \dom$ are grouped together if for all $w \in \dom$, GPTT outputs the same value for both $diff(u,w)$ and $diff(v,w)$. This results in a partitioning of the domain. Further, for every domain element $u$, we define $larger(u)$ to be the set of $v \in \dom$ such that GPTT output $\top$ for $diff(v, u)$. These are the domain elements that satisfy $x_v - x_u > \tilde{\theta}$, where $\tilde{\theta}$ is the noisy threshold. We order the partitions such that elements $u \in P_i$ have a bigger $larger(u)$ set than elements $v \in P_j$, for $j > i$. 

We can show that the ordered partitioning ${\cal P}$ imposes an ordering on the counts in the database $D$. 
\begin{lemma}
Let $D$ be a database on domain $\dom$. Let $P = \{P_1, \ldots, P_p\}$ be the ordered paritioning of $\dom$ output by Algorithm~\ref{alg:attack}. Then with probability at least $1-\delta$, for all $1 \leq \ell < m \leq p$, $u_i \in P_\ell, u_j \in P_m$, we have $x_i < x_j$.
\end{lemma}
\begin{proof}
Let $1 \leq \ell < m \leq p$, and let $\tilde{\theta}$ be the noisy threshold. Since $\theta = \lceil \frac{1}{\epsilon} \log \left(\frac{1}{\delta}\right)\rceil$, with probability at least $1-\delta$, $\tilde{\theta} > 0$. For any $u_i \in P_\ell$ and $u_j \in P_m$, $larger(u_j) \subsetneq larger(u_i)$. Therefore, there exists $u_k \in \dom$ such that $x_k - x_i  > \tilde{\theta}$, but $x_k - x_j \not> \tilde{\theta}$. Therefore, $x_i < x_j$.
\end{proof}

Let $S_i \subset \dom$ denote the set of domain elements that have count equal to $i$ in dataset $D$. It is easy to see that for every $S_i$ there is some $P \in {\cal P}$ output by Algorithm~\ref{alg:attack} such that $S_i \subseteq P$. We next show that for certain datasets there is an $m > 0$ such that the sets corresponding to small counts $0 \leq i \leq m$ are exactly reproduced in the partitioning output by Algorithm~\ref{alg:attack}. 

\begin{theorem}\label{thm:recon}
 Let ${\cal P} = \{P_0, P_1, \ldots, P_p\}$ be the ordered partition output by Algorithm~\ref{alg:attack} on $D$ with parameter $\epsilon$. Let $D$ be a dataset such that $S_i \neq \emptyset$ for all $i \in [0,k]$. That is, $D$ contains at least one domain element with count equal to $0,1,2,\ldots, k$. Let $\alpha = \lceil\frac{1}{\epsilon}\log\frac{1}{\delta}\rceil$. If $k > 2\alpha$ then with probability at least $1-\delta$, for all $i \in [0,m]$, $P_i = S_i$, where $m = k - 2\alpha$. 
\end{theorem}
\begin{proof}
Since $\theta = \alpha = \lceil\frac{1}{\epsilon}\log\frac{1}{\delta}\rceil$, with probability at least $1-\delta$, the noisy threshold $\tilde{\theta}$ will be within $[0,2\alpha]$. For the dataset $D$ such that $S_i \neq \emptyset$ for all $i \in [0,k]$, where $k >2\alpha$, for any $i \in [0,m-1]$, suppose  $u \in S_i$ and $v \in S_{i+1}$, then we have $larger(u) = \{z \mid x_z > \tilde{\theta} + i \}$ and $larger(v) = \{z \mid x_z > \tilde{\theta} + i + 1\}$.

Since $i \in [0,m]$, with probability at least $1-\delta$, $\tilde{\theta} + i \leq m + 2\alpha  < k$. Thus $\emptyset \neq S_{\lceil \tilde{\theta} + i \rceil} \subset larger(u) \setminus larger(v)$. So we have $larger(v) \subsetneq larger(u)$ and $S_i$,$S_{i+1}$ will not appear in the same $P_i \in {\cal P}$.

Furthermore, since we know $S_0, \dots S_m$ belong to separate $P_i$ and we have ${\cal P} = \{P_0, P_1, \ldots, P_p\}$ be the ordered partition. Thus, we have $P_{i} = S_{i}$ for $i \in [0,m]$.
\end{proof}

Theorem~\ref{thm:recon} shows that for datasets that have at least one domain element having a count equal to $i$ for all  $i \in [0,k]$, we can exactly tell the counts for those domain elements with count in $[0, k-2\alpha]$ with high probability.  A number of datasets satisfy the assumption that counts in $[0,k]$ all have support. For instance, for datasets that are drawn from a Zipfian distribution, the size of $S_i$ is in expectation inversely proportional to the count $i$, and thus all small counts will have support for datasets of sufficiently large size.

\begin{table}
\centering
\begin{tabular}[t]{|c|c|c|c|}
\hline
Datasets & Domain & Scale & k \\
\hline
Adult & 4096 & 17665 & 15  \\
\hline
MedicalCost & 4096 & 9415 & 26  \\
\hline
Income & 4096 & 20787122 & 98\\
\hline
HEPTH & 4096 & 347414 & 275 \\
\hline
\end{tabular}
\caption{Overview of some real world datasets}
\label{table:realdata}
\end{table}

We also find that a number of real world datasets satisfy the assumption that counts in $[0,k]$ all have support. Table~\ref{table:realdata} shows the features of some real world datasets. $\bold{Adult}$ is a histogram constructed from U.S. Census data \cite{Data:Adult} on the ``capitol loss'' attribute. $\bold{MedicalCost}$ is a histogram of personal medical expenses from the survey \cite{Data:MedCost}. $\bold{Income}$ is the histogram on ``personal income'' attribute from \cite{Data:Income}. $\bold{HEPTH}$ is a histogram constructed using the  citation network among high energy physics pre-prints on arXiv. The attributes Domain and Scale in Table~\ref{table:realdata} correspond to the size of $\dom$ and the number of tuples in the datasets. The feature $k$ for each dataset means that $S_i \neq \emptyset$ for all $i \in [0,k]$. 

However, the above attack assumes some prior knowledge about the dataset. Specifically, we assume that the attacker know $k$ such that  all counts in $[0,k]$ have support in the input dataset. Next we present an extension of our attack that allows reconstructing counts in the dataset without any prior knowledge about the dataset, but with differentially private access to the dataset.

\noindent{\bf Emperical Data Reconstruction}

\begin{algorithm}[t]
\caption{Reconstruct Algorithm}
Input: Dataset $D$ with domain $\dom$,  privacy param $\epsilon$\\
Output: A reconstruction of D
\begin{algorithmic}[1]
	\State Split the budget $\epsilon = \epsilon_1 + \epsilon_2$
	\State ${\cal P} \gets Attack~Algorithm(D, \epsilon_1)$
	\For {$P \in {\cal P}$}
		\State $\tilde{c}_P \gets count(P) + Lap(\frac{1}{\epsilon_2})$
		\State $a_P\gets \frac{\tilde{c}_P}{|P|}$
		\State We guess the count of cells in $P$ is $round(a_P)$
	\EndFor
\end{algorithmic}
\label{alg:recon}
\end{algorithm}

Algorithm~\ref{alg:recon} outlines an attack for reconstructing a dataset using GPTT and differentially private access to the dataset. Algorithm~\ref{alg:recon} also takes as input a privacy budget $\epsilon$. We split the budget $\epsilon = \epsilon_1 + \epsilon_2$. We use $\epsilon_1$ to run our GPTT based attack algorithm (Alg~\ref{alg:attack}), which outputs an ordered partition ${\cal P}$ of the domain. We use the remaining budget $\epsilon_2$ to compute noisy total counts $\tilde{c}_P$ for each partition $P \in {\cal P}$, and estimate the average count in each partition $a_P = \tilde{c}_P / |P|$.  We round this average to the nearest integer and guess that each domain element $u_i \in P$ has count $round(a_P)$.

\begin{table}
\centering
\begin{tabular}[h]{|c|c|c|c|}
\hline
 & $\epsilon=1.0$ & $\epsilon=0.5$ & $\epsilon=0.1$ \\
\hline
Adult &  0.994 & 0.991 &  0.981 \\
\hline
MedicalCost & 0.985 & 0.977 &  0.949  \\
\hline
Income & 0.798 & 0.741 & 0.636 \\
\hline
HEPTH & 0.904 & 0.795 &  0.477\\
\hline
\end{tabular}
\caption{Emperical datasets reconstruction on different $\epsilon$}
\label{table:recon}
\end{table}

We run Algorithm~\ref{alg:recon} on the four datasets. Table~\ref{table:recon} shows the fraction of domain elements in each dataset whose  counts are correctly guessed by our algorithm for $\epsilon \in \{1, 0.5, 0.1\}$. Each accuracy measure is the average of 10 repetitions. In each experiment, we set $\epsilon_1 = \epsilon_2 = 0.5\epsilon$. When $\epsilon$ is not small, the counts of most cells can be reconstructed (3 datasets can even be reconstructed  over 90$\%$ ). With the decreasing of the $\epsilon$, the ratio becomes smaller. Two reasons explain this result: (1) Small $\epsilon_1$ leads to a coarser partition from the Attack Algorithm. (2) Small $\epsilon_2$ introduces much noise to the counts of each group giving us wrong counts.

\begin{table}
\centering
\begin{tabular}[h]{|c|c|c|c|c|}
\hline
 & $\#$ of cells  &$\epsilon=1.0$ & $\epsilon=0.5$ & $\epsilon=0.1$ \\
\hline
Adult & 4062 & 0.999 & 0.997 &  0.992 \\
\hline
MedicalCost & 3878 & 1.0 & 1.0 &  0.960  \\
\hline
Income & 2369 & 1.0 & 1.0 & 0.979 \\
\hline
HEPTH & 1153 & 1.0 & 1.0 &  0.970\\
\hline
\end{tabular}
\caption{Emperical reconstruction accuracy on cells with small counts within $[0,5]$}
\label{table:recons}
\end{table}

Table~\ref{table:recons} displays the accuracy of reconstruction on domain elements with small counts within $[0,5]$. Note that more than $1/4^{th}$ of the domain has counts in $[0,5]$ for all the datasets. More than 95\% of all domain elements with small counts within $[0,5]$ can be reconstructed for all these 4 datasets under all settings of $\epsilon$ considered. Especially, when the $\epsilon$ is not small (e.g.,$\epsilon = 1.0$), nearly all these cells can be accurately reconstructed by using Algorithm~\ref{alg:recon}. 

\vspace{2ex}
\noindent{\bf Discussion:} These results show that not only does GPTT not satisfy differential privacy, it can lead to significant loss of privacy. Since  cells with small counts can be reconstructed with very high accuracy ($>95\%$), access to the data via GPTT can result in releasing query answers that can allow re-identification attacks. Hence, we believe that systems whose privacy stems from GPTT are not safe to use.

\section{Conclusion}\label{sec:conc}
We studied the privacy properties of a  variant of the sparse vector technique called generalized private threshold testing (GPTT). This technique is claimed to satisfy differential privacy and has impressive utility properties and has found applications in developing privacy preserving algorithms for frequent itemset mining, synthetic data generation and feature selection in machine learning. We show that the technique does not satisfy differential privacy. Moreover, we present attack algorithms that allows us to reconstruct counts from the input dataset (especially small counts) with high accuracy with no prior knowledge about the dataset. Thus, we demonstrate that GPTT is not a safe technique to use on datasets with privacy concerns.

\newpage
\bibliographystyle{abbrv}
\bibliography{refs}

\end{document}